\documentclass[letterpaper, 10 pt, conference]{ieeeconf}                                                             
\IEEEoverridecommandlockouts
\usepackage{graphics} 
\graphicspath{{imagenes/}}
\usepackage{epsfig} 
\usepackage{mathptmx} 
\usepackage{amsmath} 
\usepackage{amssymb}  
\usepackage{graphicx}
\usepackage{subcaption}
\usepackage{color,colortbl}
\usepackage{cleveref}
\usepackage{url}
\usepackage{enumerate}

\newtheorem{theorem}{Theorem}

\newtheorem{assumption}{Assumption}

\DeclareMathOperator{\sign}{sign}

\title{\textbf{A full controller for a fixed-wing UAV}}

\author{Gerardo Flores, Alejandro Flores and Andr\'es Montes de Oca \thanks{G. Flores, A. Flores and A. Montes de Oca are with the Perception and Robotics Laboratory, Centro de Investigaciones en \'{O}ptica, Le\'{o}n, Guanajuato, Mexico, 37150. (email: gflores@cio.mx, alejandrofl@cio.mx and andresmr@cio.mx). Corresponding author: Gerardo Flores.} 
\thanks{This work was supported partially by the FORDECYT-CONACYT under grant 292399 and by the Laboratorio Nacional de \'Optica de la Visi\'on of the CONACYT agreement 293411.}
}

\begin{document}
\maketitle
\thispagestyle{empty}
\pagestyle{empty}
\begin{abstract}
This paper presents a nonlinear control law for the stabilization of a fixed-wing UAV. Such controller solves the path-following problem and the longitudinal control problem in a single control. Furthermore, the control design is performed considering aerodynamics and state information available in the commercial autopilots with the aim of an ease implementation. It is achieved that the closed-loop system is G.A.S. and robust to external disturbances. The difference among the available controllers in the literature is: 1) it depends on available states, hence it is not required extra sensors or observers; and 2) it is possible to achieve any desired airplane state with an ease of implementation, since its design is performed keeping in mind the capability of implementation in any commercial autopilot.
\end{abstract}
\begin{keywords}
Fixed-wing; path-following; UAV; longitudinal aircraft; dynamics Lyapunov-based control.
\end{keywords}
\IEEEpeerreviewmaketitle

\section{Introduction} \label{sec:intro}
Fixed-wing Unmanned Aerial Vehicles (UAVs) have become relevant above the quadrotors in remote sensing applications, precision agriculture, and surveillance. This is due to its flight endurance, long flying times, high speeds and energy efficiency. Due to the inherent non-linearity condition of fixed-wing UAVs, effective control need to be applied to this type of vehicles so this is a topic that remains being a challenge. In this paper we propose a full controller that achieves to stabilize the airplane in a condition required to perform inspection and surveillance missions. Full control means that it can stabilize the UAV at any desired position and velocity taking into account the optimal angle of attack and aerodynamics. This controller is developed taking in mind its applicability in autopilots such as the popular Pixhawk.
%
%
\begin{figure}[ht!]
    \centering
   \begin{subfigure}{0.5\textwidth}
    \includegraphics[width=\textwidth]{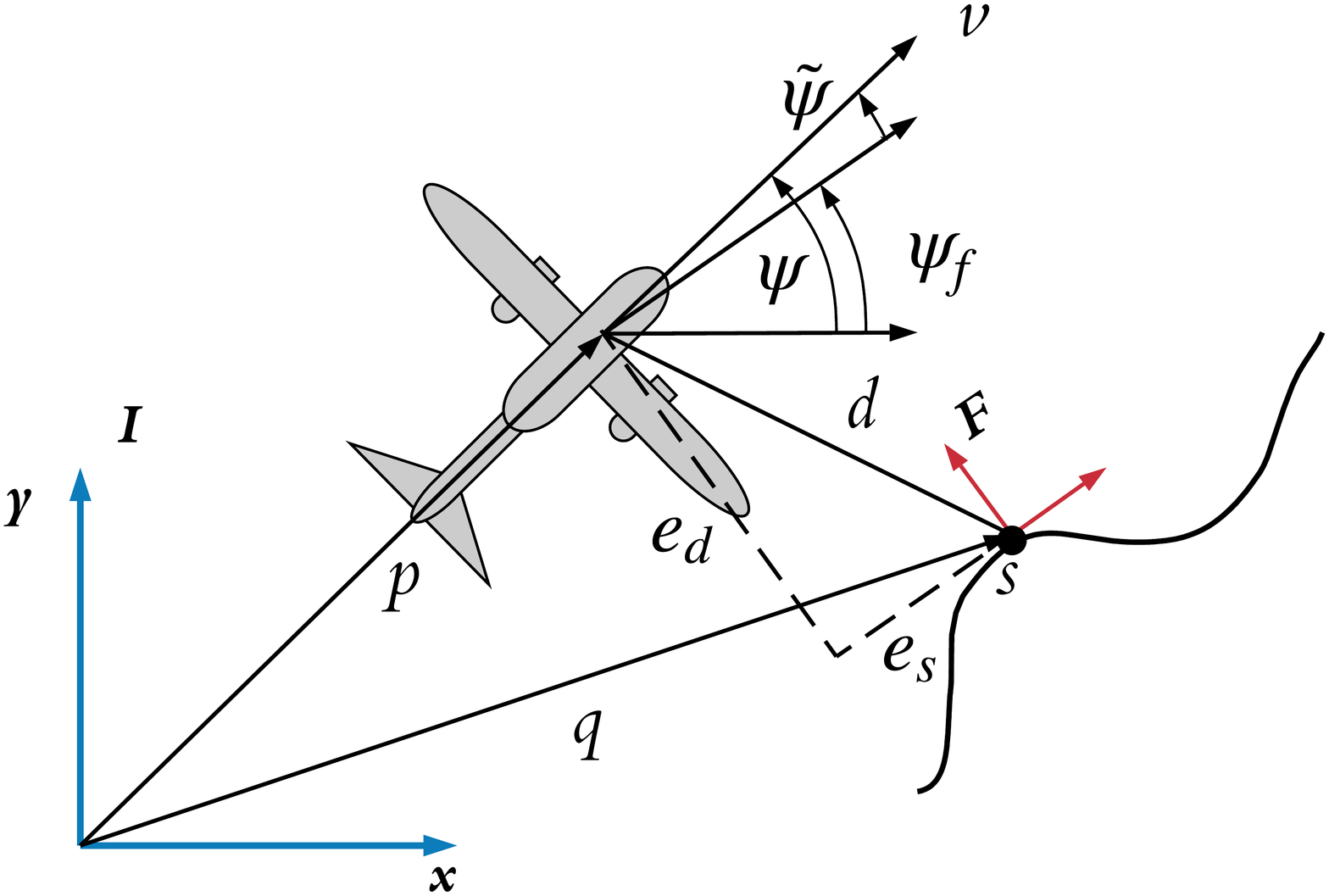}
    \caption{Lateral view.}
    \label{fig:path_following}
	\end{subfigure}
\begin{subfigure}{0.55\textwidth}
    \includegraphics[width=\textwidth]{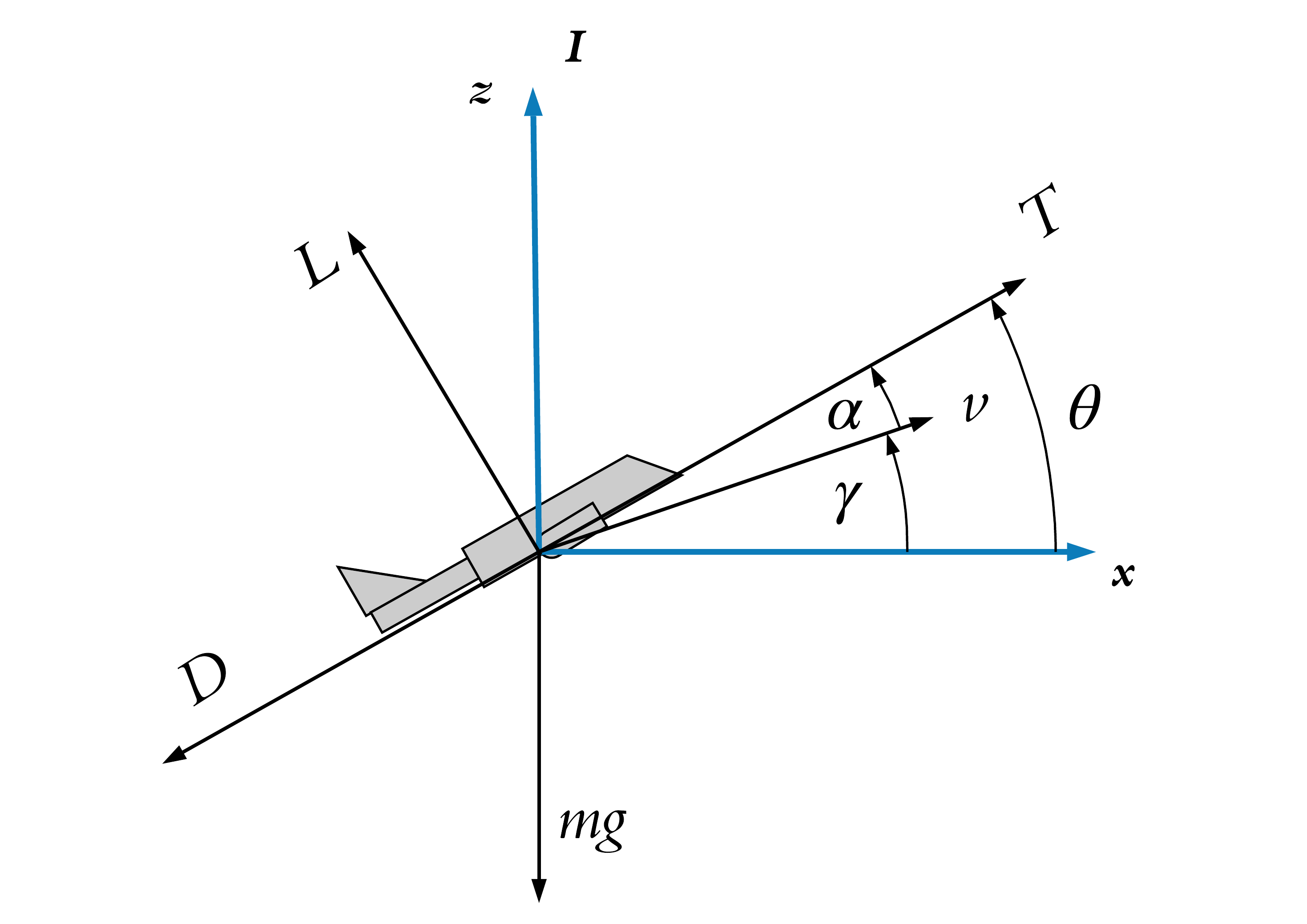}
    \caption{Longitudinal view.}
    \label{fig:longitudinal_model}
    	\end{subfigure}
    	\caption{Diagram of the fixed-wing control problem.}
    	\label{fig:models}
\end{figure}
%
\subsection{Literature review}
In the vast majority of the works relating fixed-wing aircraft control, the problem is tackled based in three fundamental modeling-based cases: a) the longitudinal model where the airspeed, angle of attack and pitch dynamics are controlled; b) path-following in 2D and 3D, where desired $(x,y)$ and $(x,y,z)$ positions must be achieved  respectively; and c) considering the full 6-DOF airplane mathematical model. Each of these approaches have its own pros and cons. Studying longitudinal model and path-following problem separately, inherits the problem of not considering the full control of the plane. Whereas in the path-following control the airspeed and aerodynamics are not considered, in the longitudinal approach the steering of the path is not taken into account. On the other hand, investigating the 6-DOF model-based control usually results in complex controls that includes a great quantity of parameters and terms that makes the implementation in real UAVs a complex task. Some of the most recent works of these three cases are shown at Table \ref{tab:biblio}.
%
\begin{table*}[ht]
\centering
\caption{State of the art of fixed-wing UAV controllers.}
\label{tab:biblio}       
\begin{tabular}{l|l|l|l|l|}
\cline{2-5}
& 6-DOF
& Path-following 2D               
& Path-following 3D     
& Longitudinal airplane dynamics control\\ \hline 
\multicolumn{1}{|l|}{No disturbance}  & \cite{hamada_receding_2018} &  \cite{zhao_integrating_2018}, \cite{jain_self-triggered_2017}, \cite{FW_pf4} & \cite{kaminer_path_2010}, \cite{vanegas_smooth_2018}, \cite{oliveira_three_2017}, \cite{cho_three-dimensional_2015}, \cite{liang_vector_2015}     &                                                                             \cite{michailidis_robust_2017}, \cite{kong_longitudinal_2014}, \cite{brigido-gonzalez_adaptive_2014}, \cite{as_controller2}\\ \hline
\multicolumn{1}{|l|}{Disturbance}     &                                        &       \cite{fortuna_cascaded_2015}                                   , \cite{liang_combined_2015} &  \cite{palframan_lpv_2015}, \cite{beard_fixed_2014}, \cite{liu_path-following_2012}  & \\ \hline
\end{tabular}
\end{table*}

Apart from the aforementioned approaches to control fixed-wing UAVs, there is the guidance law that resolves the problem of path following by a simple lateral acceleration command; this approach can be seen for instance in \cite{zhang_path-following_2014}, and in the famous paper \cite{park_new_2004} which presents the guidance law implemented in the Pixhawk commercial autopilot.

Regarding control techniques used for solving the control problem there are for instance: backstepping \cite{att_backstepping}, \cite{attitude_3}, \cite{FW_pt}; gain scheduled \cite{attitude_2}; sliding mode \cite{attitude_4}; model referenced adaptive control \cite{attitude_5}; model predictive control \cite{FW_l1_pt}, \cite{FW_guidance}; active disturbance rejection control \cite{dist_rejection_pf}; among others.
%
\subsection{Contribution}
The motivation to develop this work is twofold: a) present an algorithm that fully controls the fixed-wing UAV considering aerodynamics and the path steering problem; and b) design a controller as simple as possible that effectively stabilizes the system in the complete flight dynamics, i.e. stabilizing aircraft $(x,y)$ position, angle of attack, pitch dynamics, flight-path angle, yaw dynamics and airspeed (and as a consequence $z$ dynamics) to some given desired states. The controlled is designed having in mind its implementation in the Pixhawk autopilot using the sensor information available on it.

\subsection{Paper structure}
The rest of the paper is structured as follows. In Section \ref{sec:problem} the fixed-wing mathematical model is described, then based on desired states an error model is obtained. Section \ref{sec:result} presents the full control that stabilized the complete system, a formal proof is given. Section \ref{sec:sims} contains the simulation results that validates the effectiveness of the control. Finally, in Section \ref{sec:conclusions} some comments and future work are presented.

\section{Problem Formulation} \label{sec:problem}
\subsection{System Model}
Consider the following mathematical model representing the longitudinal and lateral dynamics according to Fig. \ref{fig:models}.
\begin{align}
&
    \begin{array}{*{20}l}
    \label{eq:model_longitudinal}
        m\dot{v} &=& T\cos{\alpha} - D - mg\sin{\gamma} \\
		mv\dot{\gamma} &=& T\sin{\alpha} + L - mg\cos{\gamma} \\
		m\dot{\theta} &=& q \\
		I_{y}\dot{q} &=& \tau \\
    \end{array}
    \\
&
    \begin{array}{*{20}l}
    \label{eq:model_lateral}
        \dot{x} &=& v \cos\psi \\
        \dot{y} &=& v \sin\psi  \\
        \dot{\psi} &=& \omega\\
    \end{array}
\end{align}
where $v$ is the airspeed of the drone, $\alpha$, $\gamma$ and $\theta$ are the angle of attack (AoA), the bank and the pitch angle, respectively, $D$ and $L$ are the drag and lift forces generated by the aerodynamics of the wing, $\dot{x}$ and $\dot{y}$ are the velocities in the inertial frame ($I$), and $\psi$ is the yaw angle. As control inputs, the dynamics can be controlled by $T$, $\tau$ and $\omega$. Without loss of generality, let normalize system \eqref{eq:model_longitudinal}, i.e. $m=I_y=1$ and then it is rewritten as
\begin{align}
    \begin{array}{*{20}l}
    \label{eq:longitudinal_simple}
        \dot{v} &=& T\cos{\alpha} - D - \sin{\gamma} \\
		v\dot{\gamma} &=& T\sin{\alpha} + L - \cos{\gamma} \\
		\dot{\theta} &=& q \\
		\dot{q} &=& \tau.
    \end{array}
\end{align}
We need to introduce the error of the angles. So, first, we define the velocity, bank angle and pitch errors as
\begin{align}
\label{longitudinal_errors}
    \begin{array}{*{20}l}
    \Tilde{v}=v-v_{d} \\
    \Tilde{\gamma}=\gamma-\gamma_{d}\\
    \Tilde{\theta}=\theta-\theta_{d}.
    \end{array}
\end{align}
Using \eqref{longitudinal_errors} in \eqref{eq:longitudinal_simple} and considering that $\Tilde{q} = q - q_{d}$, then we have
\begin{align*}
    \begin{array}{*{20}l}
        \dot{\Tilde{v}} &=& T\cos{\alpha} - D - \sin(\Tilde{\gamma} + \gamma_{d}) \\
		\dot{\Tilde{\gamma}} &=& \frac{T\sin{\alpha} + L - \cos(\Tilde{\gamma} + \gamma_{d})}{\Tilde{v} + v_{d}} \\
		\dot{\Tilde{\theta}} &=& q - q_{d}\\
		\dot{q} &=& \tau - \dot{q}_{d}.
    \end{array}
\end{align*}
If we define the following coordinate change
\begin{align*}
    \Tilde{\theta}= \begin{array}{*{20}l}
    \begin{pmatrix}
    \Tilde{\theta_{1}}\\
    \Tilde{\theta_{2}}
    \end{pmatrix} =
        \begin{pmatrix}
    \theta-\theta_{d}\\
    q-q_{d}
    \end{pmatrix}
    \end{array}
\end{align*}
and consider that $\Tilde{\alpha}=\alpha-\alpha_{d}$. The system changes to
\begin{align}
    \begin{array}{*{20}l}
    \label{eq:changed_longitudinal_model_2}
        \dot{\Tilde{v}} &=& T\cos{(\Tilde{\alpha}+\alpha_{d})} - D - \sin(\Tilde{\gamma} + \gamma_{d}) \\
		\dot{\Tilde{\gamma}} &=& \frac{T\sin{(\Tilde{\alpha}+\alpha_{d})} + L - \cos(\Tilde{\gamma} + \gamma_{d})}{\Tilde{v} + v_{d}} \\
		\dot{\Tilde{\theta}}_{1} &=& \Tilde{\theta}_{2}\\
		\dot{\Tilde{\theta}}_{2} &=& \tau - \dot{q}_{d}
    \end{array}
\end{align}
where $\dot{\Tilde{\theta}}_{2}= \Tilde{\tau}$.

For the lateral model, we define orientation and position errors according to Fig. \ref{fig:path_following}. This path following approach has been investigated in our previous work \cite{FW_pf4} where a virtual particle that moves over the desired path according to the velocity and heading of the UAV is used to achieve the following task. So, the errors are defined as follows \cite{FW_pf4}
\begin{align*}
    \begin{array}{*{20}l}
    \label{eq:error_changed_model}
    \dot{e}_{s} &=& v\cos(\Tilde{\psi}) - \dot{s} + Cc(s)e_{d}\dot{s}\\
    \dot{e}_{d} &=& v\sin(\Tilde{\psi}) - Cc(s)e_{s}\dot{s}\\
    \dot{\Tilde{\psi}} &=& \omega - Cc(s)\dot{s}
    \end{array}
 \end{align*}
where $\dot{e_{s}}$ and $\dot{e_{d}}$ are the $x$ and the $y$ errors, respectively, between the virtual particle and the position of the UAV in the inertial frame $(I)$ that is expressed in the velocity particle frame $(F)$. This rotation angle is the defined as
\begin{equation*}
     \psi_{f} = \arctan\frac{y'_{s}}{x'_{s}}
\end{equation*}
It is also introduced the path curvature of the desired trajectory as $C_c$. As $v=\Tilde{v}+v_{d}$, the previous system now is
\begin{align}
    \begin{array}{*{20}l}
\label{eq:error_changed_model_2}
    \dot{e}_{s} &=& (\Tilde{v}+v_{d})\cos(\Tilde{\psi}) - \dot{s} + Cc(s)e_{d}\dot{s}\\
    \dot{e}_{d} &=& (\Tilde{v}+v_{d})\sin(\Tilde{\psi}) - Cc(s)e_{s}\dot{s}\\
    \dot{\Tilde{\psi}} &=& \omega - Cc(s)\dot{s}
\end{array}
\end{align}
Gathering \eqref{eq:changed_longitudinal_model_2} and \eqref{eq:error_changed_model_2}, the complete error model is given by
\begin{align}
&
    \begin{array}{*{20}l}
    \label{eq:final_model}
        \dot{\Tilde{v}} &=& T\cos{(\Tilde{\alpha}+\alpha_{d})} - D - \sin(\Tilde{\gamma} + \gamma_{d}) \\
		\dot{\Tilde{\gamma}} &=& \frac{T\sin{(\Tilde{\alpha}+\alpha_{d})} + L - \cos(\Tilde{\gamma} + \gamma_{d})}{\Tilde{v} + v_{d}} \\
		\dot{\Tilde{\theta}}_{1} &=& \Tilde{\theta_{2}}\\
		\dot{\Tilde{\theta}}_{2} &=&\Tilde{\tau}\\
        \dot{e}_{s} &=& v\cos(\Tilde{\psi}) - \dot{s} + Cc(s)e_{d}\dot{s}\\
        \dot{e}_{d} &=& v\sin(\Tilde{\psi}) - Cc(s)e_{s}\dot{s}\\
        \dot{\Tilde{\psi}} &=& \omega - Cc(s)\dot{s}.
    \end{array}
\end{align}
\section{Main Result} \label{sec:result}
In this section the main result is resented. For that, let consider the following assumptions.
\begin{assumption}
The airspeed velocity is always positive, i.e. $v \in (0,c]$ with $c \in \mathbb{R}^+$.
\end{assumption}
\begin{assumption}
$v_{d}$ and $\gamma_{d}$ are considerd constants.
\end{assumption}
\begin{assumption}
The initial airspeed is positive, i.e., the fixed-wing UAV has taken off and is flying in the air.
\end{assumption}

The main results is presented in the next
\begin{theorem}
Let system \eqref{eq:final_model} under controls \eqref{eq:T_control}, \eqref{eq:control_tau}, \eqref{eq:sdot} and \eqref{eq:control_omega} then the closed loop system is globally asymptotically stable.
\end{theorem}
\begin{proof}
The goal is to get $v \rightarrow v_d$ and $\alpha \rightarrow \alpha_d$ with $v_d, \alpha_d \in \mathbb{R}^+$. It is important to highlight that the algebraic equation $\theta = \gamma + \alpha$ must be always fulfilled. This is a condition inherent of the AoA, pitch and flight-path angle.
Let consider the following $C^1$ function
\begin{equation}
\label{eq:sum_v1}
   {\Tilde{V}}_{1}= {\Tilde{V}}_{\Tilde{v}}+{\Tilde{V}}_{\Tilde{\gamma}}+{\Tilde{V}}_{\Tilde{\theta}}
\end{equation}
where
    \begin{equation}\label{eq:V1s}
    \Tilde{V}_{\Tilde{v}}=\frac{1}{2}\Tilde{v}^2; \hspace{5mm} \Tilde{V}_{\Tilde{\gamma}}=\frac{1}{2}\Tilde{\gamma}^2; \hspace{5mm} \Tilde{V}_{\Tilde{\theta}}=\Tilde{\theta}^{T}P\Tilde{\theta} 
\end{equation}
with $P = P^{T} > 0$. The time derivative of $\Tilde{V}_{\Tilde{v}}$ is
\begin{align*}
    \begin{array}{*{20}l}
   \dot{\Tilde{V}}_{\Tilde{v}}= \Tilde{v} \dot{\Tilde{v}}=
      \Tilde{v} \left(T\cos{(\Tilde{\alpha}+\alpha_{d})} - D - \sin(\Tilde{\gamma} + \gamma_{d})\right)
\end{array}
\end{align*}
where we have designed the control term $T\cos{(\Tilde{\alpha}+\alpha_{d})}$ as
\begin{align}
    \begin{array}{*{20}l}
     T\cos{(\Tilde{\alpha}+\alpha_{d})} =D + \sin(\Tilde{\gamma} + \gamma_{d})-k_{v}\Tilde{v}
     \label{eq:control_T1}
\end{array}
\end{align}
and then under this control the time derivative of $\Tilde{V}_{\Tilde{v}}$ results in
\begin{align}
    \begin{array}{*{20}l}\label{eq:vv}
     \dot{\Tilde{V}}_{\Tilde{v}}= -k_{v}\Tilde{v}^{2}.
\end{array}
\end{align}
Now obtaining the time derivative of $\Tilde{V}_{\Tilde{\gamma}}$ we obtain
\begin{align*}
    \begin{array}{*{20}l}
     \dot{\Tilde{V}}_{\Tilde{\gamma}}= \Tilde{\gamma} \left(T\sin{(\Tilde{\alpha}+\alpha_{d})} + L - \cos(\Tilde{\gamma} + \gamma_{d})\right)\frac{1}{\Tilde{V}+V_{d}}
\end{array}
\end{align*}
where we choose the control term $T\sin{(\Tilde{\alpha}+\alpha_{d})}$
equal to
\begin{align}
    \begin{array}{*{20}l}
    T\sin{(\Tilde{\alpha}+\alpha_{d})} = -L + \cos(\Tilde{\gamma} + \gamma_{d})-(\Tilde{v}+v_{d})(k_{\gamma}\Tilde{\gamma})
    \label{eq:control_T2}
\end{array}
\end{align}
and therefore
\begin{align}
    \begin{array}{*{20}l}\label{eq:vgamma}
     \dot{\Tilde{V}}_{\Tilde{\gamma}}= -k_{\gamma}\Tilde{\gamma}^{2}.
\end{array}
\end{align}
Considering $u_{1} = T\cos{(\Tilde{\alpha}+\alpha_{d})}$ and $u_{2} = T\sin{(\Tilde{\alpha}+\alpha_{d})}$ as temporal control inputs, we can extract $T$ as
\begin{equation}
    \label{eq:T_control}
    T = \sqrt{u_{1}^{2} + u_{2}^{2}}
\end{equation}
$\alpha_{d}$ is obtained from $u_{2} = T\sin\alpha$ obtaining $\alpha$ as follows
\begin{equation}
    \label{eq:alpha_comand}
    \alpha_{d} = \arcsin \left( \frac{u_{2}}{T} \right),
\end{equation}
this $\alpha_{d}$ is the commanded for $\theta_{d} = \alpha_{d} + \gamma_{d}$. The pitch error subsystem can be represented by
\begin{align*}
    \begin{array}{*{20}l}
     \dot{\Tilde{\theta}}= A\Tilde{\theta}+B\Tilde{\tau} \hspace{5mm} \text{with} \hspace{5mm} A=
        \begin{pmatrix}
    0 & 1\\
    0 & 0
    \end{pmatrix}; \hspace{5mm}
        B=
        \begin{pmatrix}
    0 \\
    1
    \end{pmatrix}.
\end{array}
\end{align*}
Let 
\begin{align}
    \begin{array}{*{20}l}
     \Tilde{\tau}=-k_{\Tilde{\theta_{1}}}\Tilde{\theta_{1}}-k_{\Tilde{\theta}_{2}}\Tilde{\theta_{2}}
     \label{eq:control_tau}
\end{array}
\end{align}
then $\dot{\Tilde{\theta}}= \Tilde{A}\Tilde{\theta}$ where $\Tilde{A} = \begin{pmatrix}
    0 & 1\\
    -k_{1} & -k_{2}
    \end{pmatrix}.$
It follows that
\begin{equation}
\label{eq:vtheta}
     \dot{\Tilde{V}}_{\Tilde{\theta}}=\Tilde{\theta}^{T}(P\Tilde{A}+ \Tilde{A}^{T}P) \Tilde{\theta}=-\Tilde{\theta}^{T}Q\Tilde{\theta}
\end{equation}
with the Lyapunov equation $P\Tilde{A}+ \Tilde{A}^{T}P=-Q$ with $Q=Q^{T}>0$.
Then, it follows that the time derivative of \eqref{eq:sum_v1} is given by 
\begin{align}
    \begin{array}{*{20}l}\label{eq:sumv1}
   \dot{\Tilde{V}}_{1}=-k_{v}\Tilde{v}^{2}-k_{\gamma}\Tilde{\gamma}^{2}-\Tilde{\theta}^{T}Q\Tilde{\theta}\leq 0
\end{array}
\end{align}

From the previous analysis one concludes that all trajectories $E(t) = (\Tilde{v}(t), \Tilde{\gamma}(t), \Tilde{\theta_{1}}(t), \Tilde{\theta_{2}}(t), \Tilde{e_{s}}(t), \Tilde{e_{d}}(t), \Tilde{\psi}(t))$ are bounded, i.e. for each trajectory $E(t)$ there is an $R \in \mathbb{R}$ such that $||E(t)|| \leq R \hspace{1mm} \forall \hspace{1mm} t \geq 0$.

It is known that any bounded trajectory of an autonomous system like \eqref{eq:final_model}, converges to an invariant set \cite{slotine}.
Let such system invariant set be $\Omega_{l}$. By construction we can say that
\begin{align}
    \begin{array}{*{20}l}\label{eq:IS}
    \Omega_{l}=\{ E(t) \in \mathbb{R}^{7}:\Tilde{V}(E)\leq l \}
\end{array}
\end{align}
where $\Tilde{V}(E)$ is a Lyapunov function for \eqref{eq:final_model}. Let find $\Tilde{V}(E)$ by combining $\Tilde{V}_1$, which is not dependent of every states in \eqref{eq:final_model}, with the positive definite function
\begin{align}
    \begin{array}{*{20}l}\label{eq:sum_v2}
   {\Tilde{V}}_{2}= {\Tilde{V}}_{\Tilde{e}_{s}}+{\Tilde{V}}_{\Tilde{e}_{d}}+{\Tilde{V}}_{\Tilde{\psi}}
\end{array}
\end{align}
where
    \begin{equation}\label{eq:V2s}
\Tilde{V}_{\Tilde{e}_{s}}=\frac{1}{2}e_{s}^{2};\hspace{5mm}\Tilde{V}_{\Tilde{e}_{d}}=\frac{1}{2}e_{d}^{2};\hspace{5mm}\Tilde{V}_{\Tilde{\psi}}=\frac{1}{2}(\Tilde{\psi}-\delta(e_{d}))^{2}
\end{equation}
where $\delta(e_{d})$ is a function which is saturated as in \cite{saturation} that satisfies the condition $e_{d}\delta(e_{d}) \leq 0 \hspace{1mm} \forall \hspace{1mm} e_{d}$. Let's compute the time derivative of $\Tilde{V}_{e_{s}}+\Tilde{V}_{e_{d}}$ and then we get
\begin{align*}
    \begin{split}
    \dot{\Tilde{V}}_{e_{s}}+\dot{\Tilde{V}}_{e_{d}}=e_{s}\dot{e}_{s}+e_{d}\dot{e}_{d}=  e_{s}[(& \Tilde{v}  + v_{d})\cos{\Tilde{\psi}}  - \dot{s}+Cc(s)e_{d}\dot{s}]\\
    & + e_{d}[(\Tilde{v}+v_{d})\sin{\Tilde{\psi}}-Cc(s)e_{s}\dot{s}].
    \end{split}
\end{align*}
If we define $\dot{s}$ as
\begin{equation}\label{eq:sdot}
    \dot{s} = k_{s} \sign(e_{s})+(\Tilde{v}+v_{d})\cos{\Tilde{\psi}}
\end{equation}
then
\begin{align}
    \begin{array}{*{20}l}\label{eq:sum_ves_ved}
    \dot{\Tilde{V}}_{e_{s}}+\dot{\Tilde{V}}_{e_{d}}= -k_{s}|e_{s}|+(\Tilde{v}+v_{d})e_{d}\sin{\Tilde{\psi}}.
    \end{array}
\end{align}
Now if we obtain the time derivative of $\Tilde{V}_{\Tilde{\psi}}$
\begin{align*}
    \begin{split}
    \dot{\Tilde{V}}_{\psi}= (\Tilde{\psi}-\delta(e_{d}))(\omega  &- Cc(s)\dot{s}\\
    &-\dot{\delta}[(\Tilde{v}+v_{d})\sin{\Tilde{\psi}} - Cc(s) e_{s}\dot{s}])
    \end{split}
\end{align*}
then consider the following proposed control
\begin{align}
    \begin{split}
    \omega= Cc(s)\dot{s} &+\dot{\delta}[(\Tilde{v}+v_{d})\sin{\Tilde{\psi}} - Cc(s) e_{s}\dot{s}]]\\
    &-(\Tilde{v}+v_{d})e_{d}\left(\frac{\sin{\Tilde{\psi}} -\sin{\delta}}{\Tilde{\psi}-\delta}\right) - k_{\omega} \sign (\Tilde{\psi}-\delta)
    \label{eq:control_omega}
    \end{split}
\end{align}
then
\begin{align}
    \begin{array}{*{20}l}\label{eq:v_psi}
    \dot{\Tilde{V}}_{\psi}= -(\Tilde{v}+v_{d})e_{d}\sin{\Tilde{\psi}}+(\Tilde{v}+v_{d})e_{d}\sin{\delta}-k_{\omega}|\Tilde{\psi}-\delta(ed)|.
    \end{array}
\end{align}
Using \eqref{eq:sum_ves_ved} and \eqref{eq:v_psi} in \eqref{eq:sum_v2}, it follows that
\begin{align}
    \begin{array}{*{20}l}\label{eq:sumv2}
   \dot{\Tilde{V}}_{2}=  -k_{s}|e_{s}|+\Tilde{v}e_{d}\sin{\delta(e_{d})} +v_{d}e_{d}\sin{\delta(e_{d})}\\
   \hspace{7mm}-k_{\omega}|\Tilde{\psi}-\delta(ed)|.
\end{array}
\end{align}
Then, as 
\begin{align*}
    \begin{array}{*{20}l}
   \dot{\Tilde{V}}= \dot{\Tilde{V}}_{1}+\dot{\Tilde{V}_{2}} 
\end{array}
\end{align*}
we substitute \eqref{eq:sumv1} and \eqref{eq:sumv2}
\begin{align}
    \begin{array}{*{20}l}\label{eq:v_sum}
   \dot{\Tilde{V}}=-k_{v}\Tilde{v}^{2}-k_{\gamma}\Tilde{\gamma}^{2}-\Tilde{\theta}^{T}Q\Tilde{\theta} -k_{s}|e_{s}|+\Tilde{v}e_{d}\sin{\delta(e_{d})} \\ 
   \hspace{7mm}+v_{d}e_{d}\sin{\delta(e_{d})}-k_{\omega}|\Tilde{\psi}-\delta(e_{d})|
\end{array}
\end{align}
From \eqref{eq:IS} it follows that
\begin{align}
    \begin{array}{*{20}l}\label{eq:2L}
   |\Tilde{V}|=\sqrt{2l}.
\end{array}
\end{align}
Then if we define $a$ from \eqref{eq:v_sum} as
\begin{align}
    \begin{array}{*{20}l}\label{eq:a}
   a=-k_{v}\Tilde{v}^{2}-k_{\gamma}\Tilde{\gamma}^{2}-\Tilde{\theta}^{T}Q\Tilde{\theta} -k_{s}|e_{s}|-k_{\omega}|\Tilde{\psi}-\delta(e_{d})|
\end{array}
\end{align}
and we substitute \eqref{eq:a} and \eqref{eq:2L} in \eqref{eq:v_sum} we have
\begin{align*}
    \begin{array}{*{20}l}
   \dot{\Tilde{V}}\leq a +\sqrt{2l}|e_{d}\sin{\delta(e_{d})}|+v_{d}e_{d}\sin{\delta(e_{d})}
\end{array}
\end{align*}
and
\begin{align*}
    \begin{split}
   a +\sqrt{2l}|e_{d}\sin{\delta(e_{d})}|&+v_{d}e_{d}\sin{\delta(e_{d})} \\
   & \leq a -(v_{d}-\sqrt{2l})|e_{d}\sin{\delta(e_{d})}|
\end{split}
\end{align*}
since $-|e_{d}\sin{\delta(e_{d})}|=e_{d}\sin{\delta(e_{d})}$, then $\dot{\Tilde{V}} \leq 0$ if $v_{d}>\sqrt{2l}$. Also, if $\psi=\delta(e_{d})$ and the remaining states are zero, $\dot{\Tilde{V}}=0$. Then it is clear that $\Tilde{V}>0 \hspace{1mm} \forall \hspace{1mm} E(t) \neq 0$ and $\dot{\Tilde{V}}(E(t)) \leq 0 \hspace{1mm} \forall \hspace{1mm} E(t)$ in $\Omega_{l}$. Let $S \subseteq \Omega_{l}$ the set of all points where $\dot{\Tilde{V}}=0$ given by
\begin{align*}
    \begin{array}{*{20}l}
   \dot{\Tilde{V}} = 0 = -k_{\omega}|\Tilde{\psi}-\delta(e_{d})| \rightarrow \Tilde{\psi}=\delta(e_{d}).
   \end{array}
\end{align*}
Note that the rest of $\dot{\Tilde{V}}$ is equal to zero at zero. Observing the closed loop system \eqref{eq:final_model} with the controls \eqref{eq:T_control}, \eqref{eq:control_tau} and \eqref{eq:control_omega} when $\Tilde{v} = \Tilde{\gamma} = \Tilde{\theta}_{1} = \Tilde{\theta}_{2} = e_{s} = 0$ and $\Tilde{\psi} = \delta(e_{d})$ one obtains
\begin{align}
    \begin{array}{*{20}l}
    \dot{e}_{s} &=& 0 = v_{d}\cos{\Tilde{\psi}} - v_{d}\cos{\Tilde{\psi}} + Cc(s)e_{d}[v_{d}\cos{\Tilde{\psi}}]\\
    &=& Cc(s)v_{d}e_{d}\cos{\delta(e_{d})}\\
   \dot{e}_{d} &=& 0 = v_{d}\sin{\Tilde{\psi}}\\
   \dot{\Tilde{\psi}} &=& \dot{\delta}(v_{d}\sin{\Tilde{\psi}})
   \label{eq:es_ed_psi_errors}
\end{array}
\end{align}
then $e_{d} = \psi = 0$ is the only solution of \eqref{eq:es_ed_psi_errors} and therefore the largest invariant set $M$ in $S$ is the origin $E(t) = 0$ and by the LaSalle invariance principle $E(t) \rightarrow 0$ as $t \rightarrow \infty$; as long as the initial state $E_{0}(t)$ be in $\Omega_{e}$. And the closed loop system is asymptotically stable in the arbitrary large set: $\Omega_{e}$ provided that $\sqrt{2l} \leq v_{d}$. Also $\tilde{V} \rightarrow \infty$ as $|E(t)| \rightarrow \infty$, i.e. $\Tilde{V}$ is radially unbounded, then the closed loop system is G.A.S.
\end{proof}
\section{Simulation Experiments} \label{sec:sims}
To verify the stability in a computational environment, a numerical simulation of the systems \eqref{eq:model_longitudinal} and \eqref{eq:model_lateral} under controls \eqref{eq:T_control}, \eqref{eq:control_tau}, \eqref{eq:sdot} and \eqref{eq:control_omega} is performed. The parameters used in the simulation are: $\psi_{a} = \frac{\pi}{2.1}$, $m = 3$, $ g= 9.81$, $I_{y} = 1$, $k_{\tilde{\theta}_{1}} = 10$, $k_{\tilde{\theta}_{2}} = 5$, $k_{v} = 50$, $k_{s} = 2$,  $k_{\gamma} = 2$, $k_{\omega} = 20.25$, $\gamma_{d} = 1^{\circ}$, $V_{d} = \sqrt{\frac{m \cdot g}{.96*c}}$, where $c = \frac{pS}{2}$ is a constant that describes the air density an the area of the wind. The initial conditions for the simulation are: $s_{0} = 0$, $x_{0} = 10$, $y_0 = -5$, $\psi_{0} = 0^{\circ}$, $\gamma_0 = 8^{\circ}$, $\theta_{0} = 8^{\circ}$, $V_{0} = 1$. The induced disturbance to the $\psi$ state is a sinusoidal signal with $20$ units of amplitude and a frequency of $0.5 \frac{rad}{s}$. The airspeed is depicted at Fig. \ref{fig:velocity_state}. The control $T$ applied to the model is shown in Fig. \ref{fig:trhust_input}.
\begin{figure}[ht]
    \centering
    \includegraphics[width=\columnwidth]{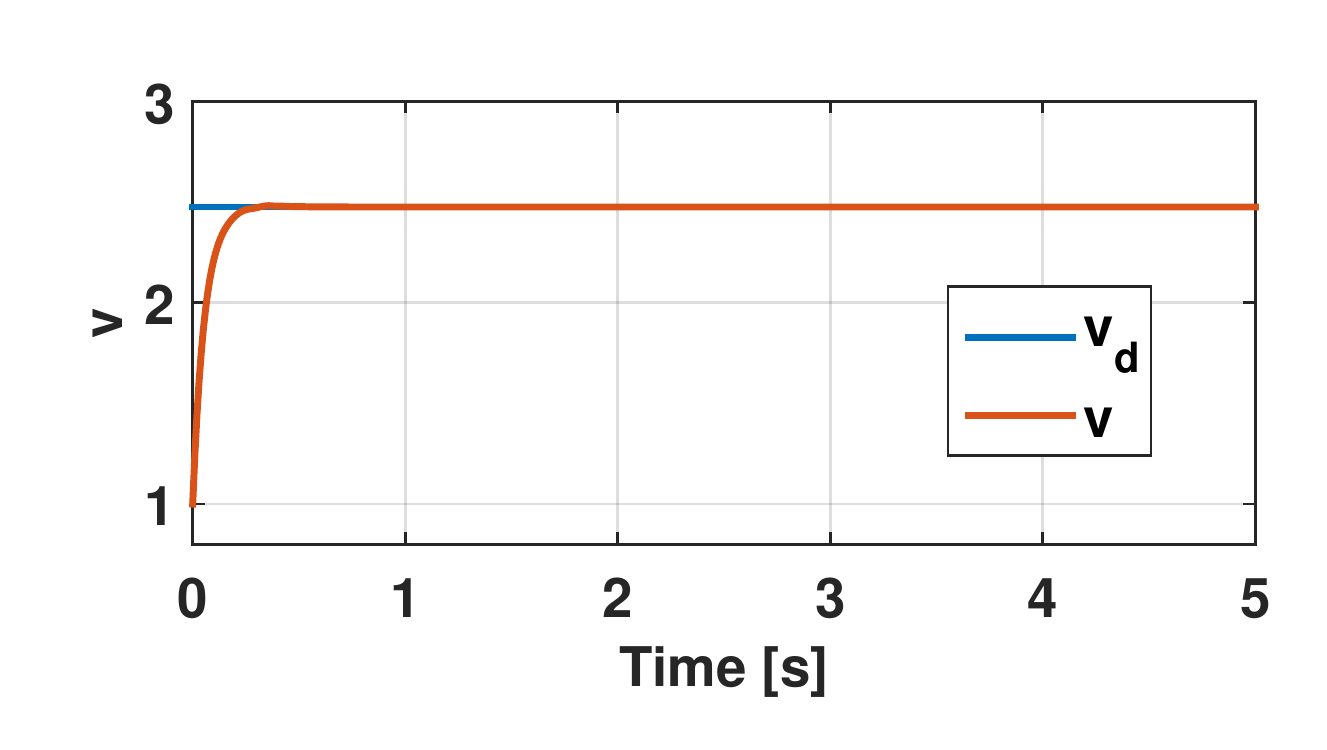}
    \caption{Airspeed $v$ through the UAV control for the optimal AoA considering aerodynamics of \cite{8619303}.}
    \label{fig:velocity_state}
\end{figure}
\begin{figure}[ht]
    \centering
    \includegraphics[width=\columnwidth]{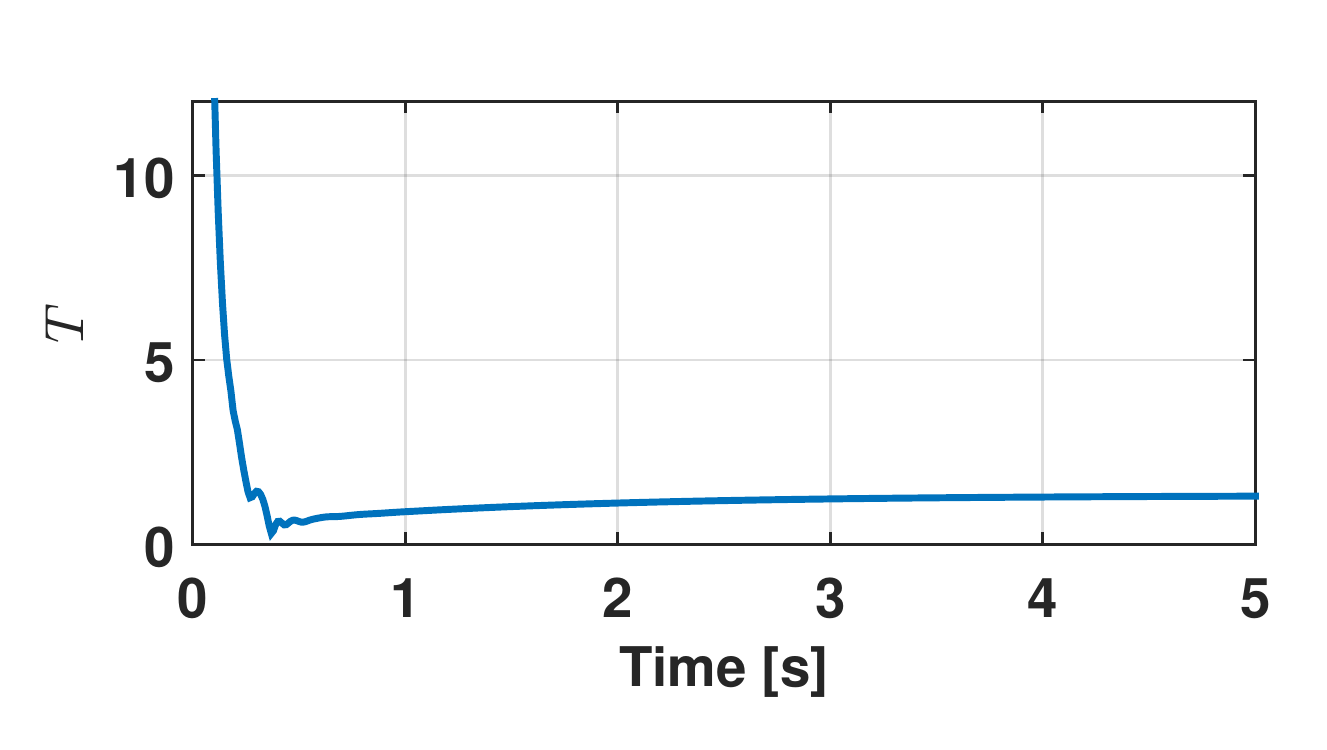}
    \caption{Thrust control $T$ applied to the fixed-wing UAV.}
    \label{fig:trhust_input}
\end{figure}

Another important fact that determines an optimal flight of the UAV is the AoA. According to the lift coefficient of the wing profile used for this simulation, which is a symmetric profile, the best AoA in relation with the lift-drag coefficients is around $6^{\circ}$. Then, it is expected that this angle is reached during stabilization process. This AoA and the other aerodynamic terms have been taken from the UAV presented in our previous work \cite{8619303}. In Fig. \ref{fig:angles} it is shown the angles evolution including the AoA convergence to the desired value equal to $6^{\circ}$.
\begin{figure}[ht]
    \centering
    \includegraphics[width=\columnwidth]{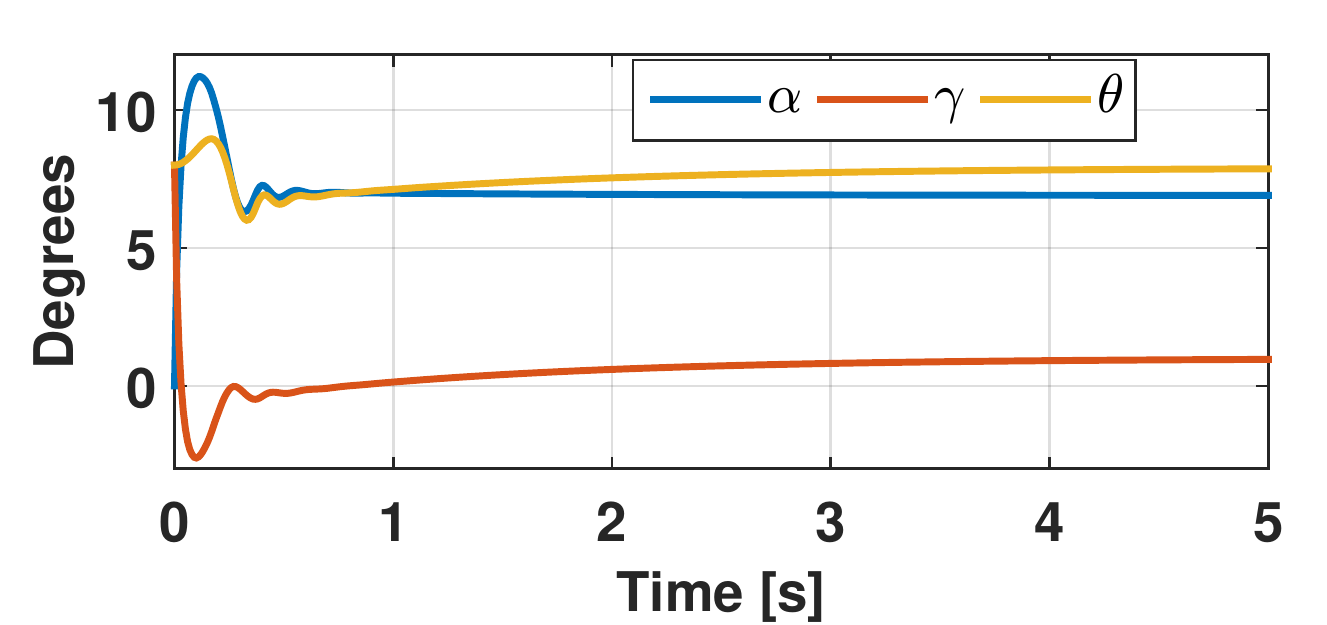}
    \caption{UAV angles $\alpha$, $\gamma$ and $\theta$ keeping the property of $\theta = \alpha + \gamma$.}
    \label{fig:angles}
\end{figure}
The control input $\tau$ is shown at Fig. \ref{fig:tau}. With this control it is possible to reach the $\theta_{d}$ according to $\gamma_{d}$ and $\alpha_{d}$.
\begin{figure}[ht]
    \centering
    \includegraphics[width=\columnwidth]{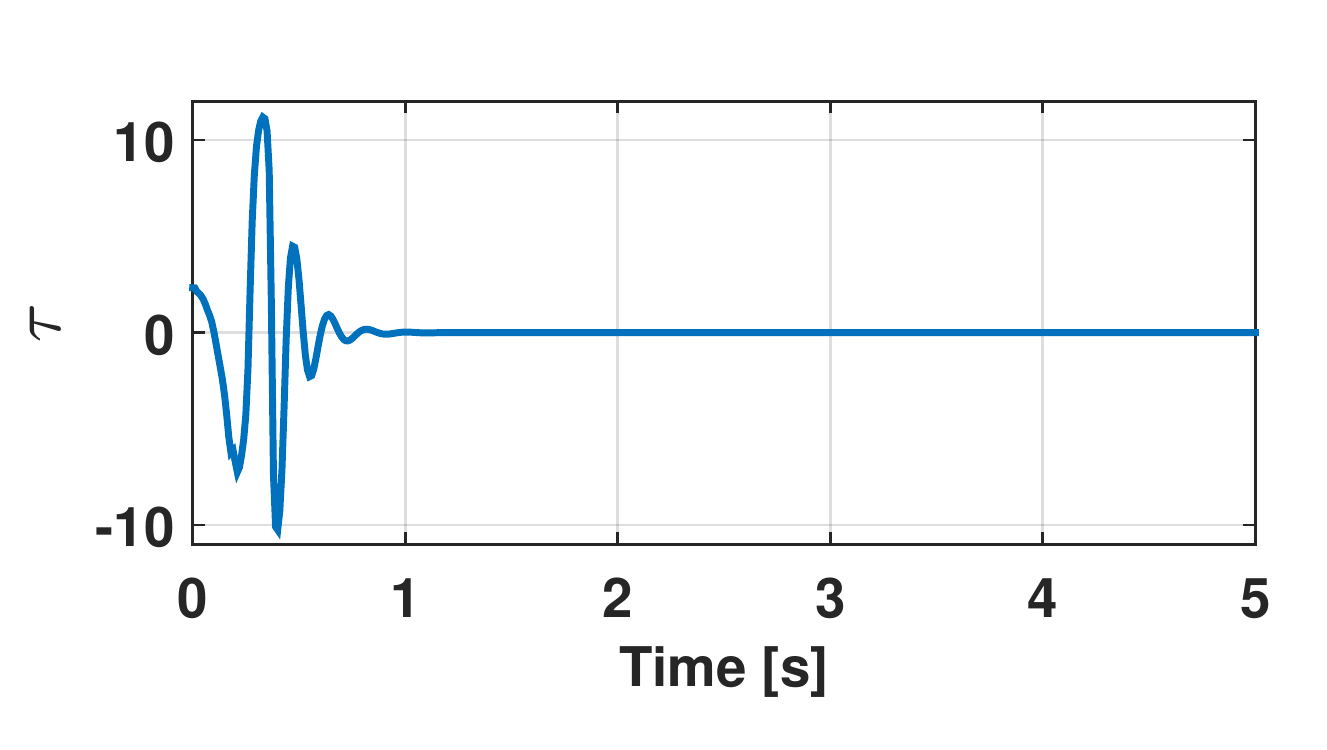}
    \caption{Control $\tau$ for the UAV pitch angle.}
    \label{fig:tau}
\end{figure}

In Fig.\ref{fig:omega}, the control $\omega$ is plotted. With this control, $\psi$ is guided to the desired heading for the path following.
\begin{figure}
    \centering
    \includegraphics[width=\columnwidth]{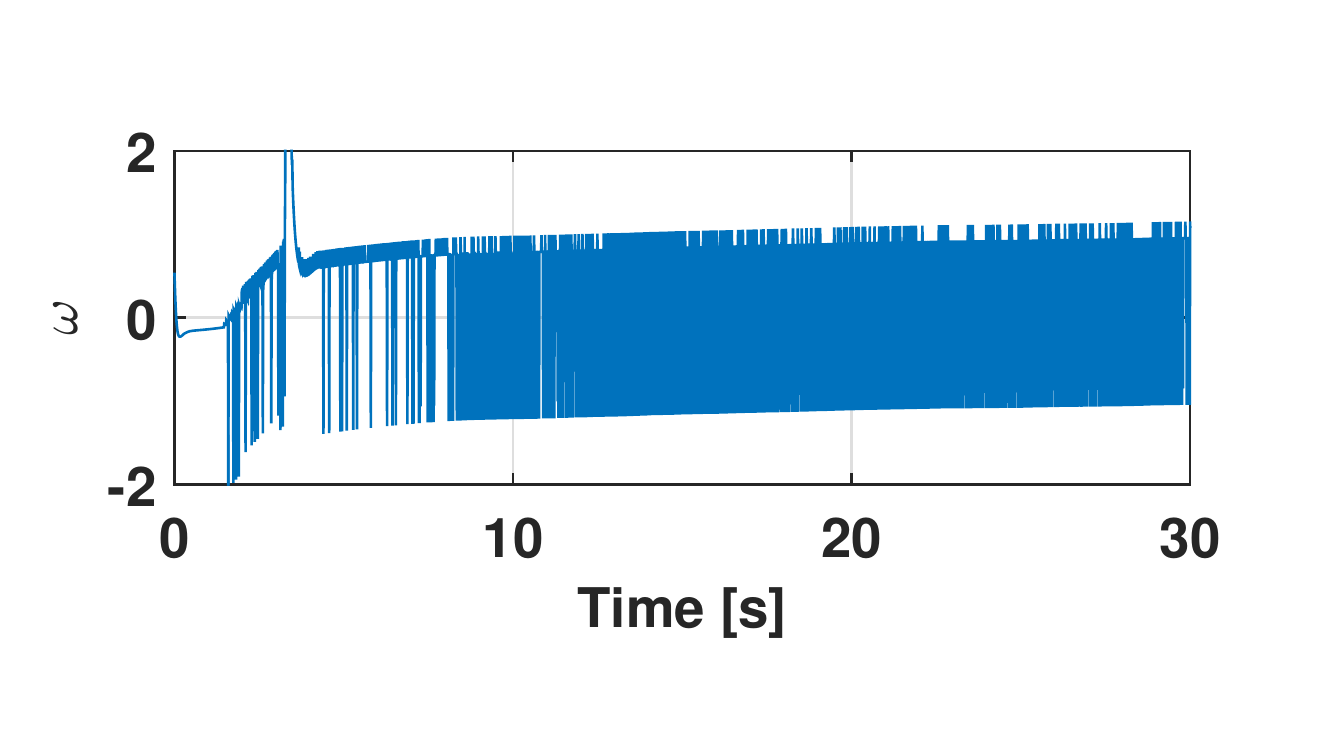}
    \caption{$\omega$ control applied to the $\psi$ state in the system.}
    \label{fig:omega}
\end{figure}
Finally, Figs. \ref{fig:path} and \ref{fig:path_np} depict the UAV's position w.r.t. the desired path, which is in this case a circle of radius $20$ with and without disturbance in $\psi$, respectively.
\begin{figure}
    \centering
    \includegraphics[width=\columnwidth]{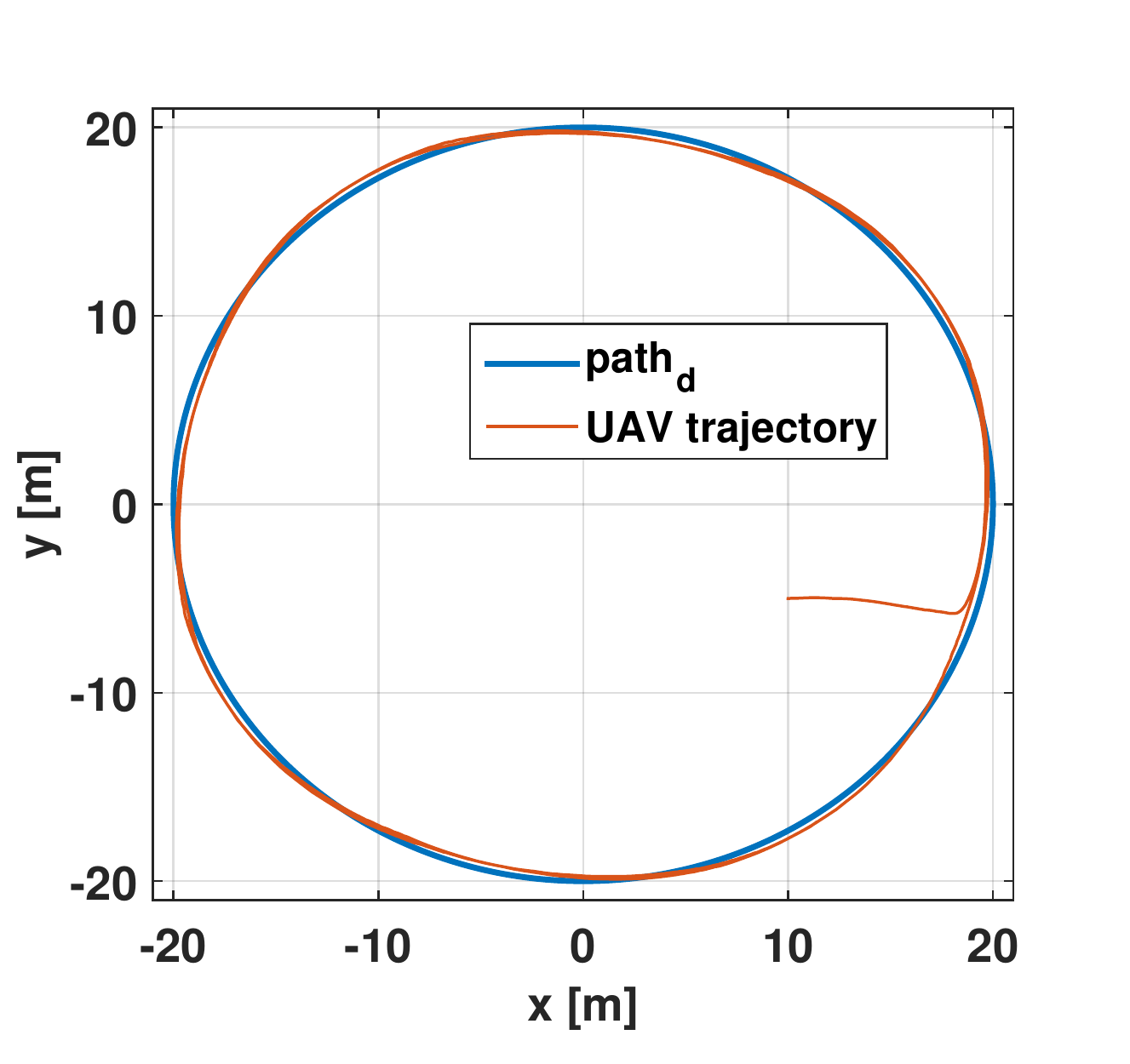}
    \caption{UAV $(x,y)$ position during the simulation; and the desired path with disturbance in $\psi$.}
    \label{fig:path}
\end{figure}
\begin{figure}
    \centering
    \includegraphics[width=\columnwidth]{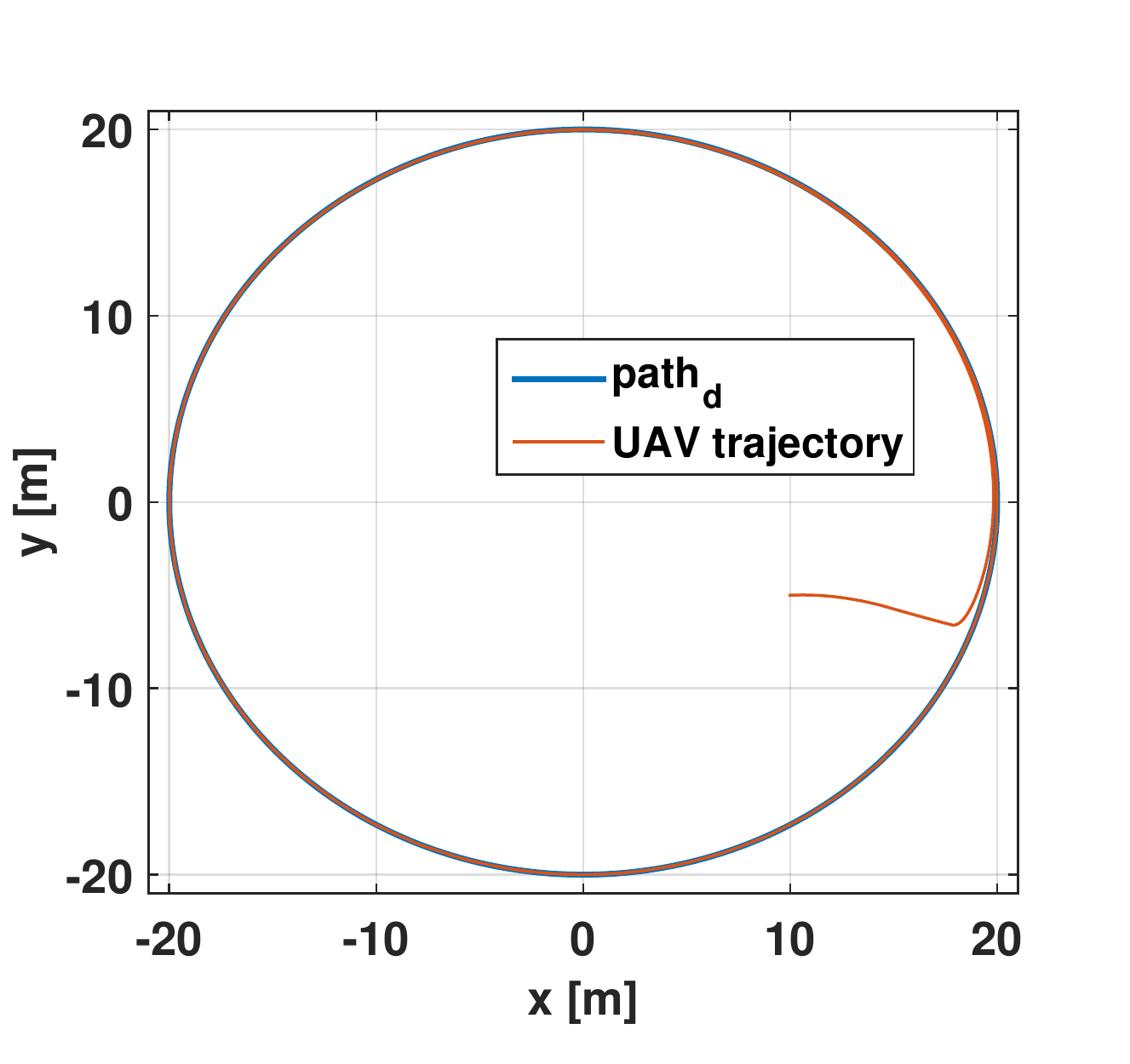}
    \caption{UAV $(x,y)$ position during the simulation; and the desired path without disturbance.}
    \label{fig:path_np}
\end{figure}
\section{Conclusions} \label{sec:conclusions}
We have presented a full controller for a fixed-wing UAV which presents robust capabilities for disturbances presented in the UAV dynamics. A stability proof is presented which demonstrated that the closed-loop system is GAS.

As future work we intend to implement this controller used the approach known as hardware in the loop (HIL) and finally be implemented in the real fixed-wing UAV in a Pixhawk autopilot.
\bibliographystyle{IEEEtran}
\bibliography{biblio}
\end{document}